\documentclass[aps,pra,notitlepage,10pt,a4paper]{revtex4-1}
\usepackage[utf8]{inputenc}
\usepackage[sc,osf]{mathpazo}
\usepackage[T1]{fontenc}
\usepackage{amsfonts}
\usepackage{amssymb}
\usepackage{amsmath}
\usepackage{amsthm}
\usepackage{graphics}
\usepackage{graphicx}
\usepackage{braket}
\usepackage[colorlinks=true,linkcolor=blue,
urlcolor=blue,citecolor=blue]{hyperref}
\newtheorem{theorem}{Theorem}

\begin{document}
\title{Locally distinguishing quantum states with limited classical communication}

\author{Saronath Halder}
\email{saronath.halder@gmail.com}
\affiliation{Quantum Information and Computation Group, Harish-Chandra Research Institute, HBNI, Chhatnag Road, Jhunsi, Prayagraj (Allahabad) 211 019, India}

\author{Chirag Srivastava}
\email{chiragsrivastava@hri.res.in}
\affiliation{Quantum Information and Computation Group, Harish-Chandra Research Institute, HBNI, Chhatnag Road, Jhunsi, Prayagraj (Allahabad) 211 019, India}

\begin{abstract}
We consider different settings of the task to distinguish pure orthogonal quantum states under local operations and a limited amount of classical communication. In the first setting, the spatially separated parties are allowed to perform only local projective measurements without any classical communication during the measurements. Under such a restricted class of operations, if the states are indistinguishable, then within a second group of settings, the parties are allowed to use an additional resource for the distinguishing. The additional resource is either a pure entangled state or multiple identical copies of the given states. Comparisons between these two types of resources are also done in certain cases. Both probabilistic and perfect discrimination of the states are considered for the second group of settings. Within a third setting, the parties perform local projective measurements with a restriction on the availability of classical communication during the measurements. But in this setting the parties are allowed to use a maximally entangled state as a resource. 
\end{abstract}
\maketitle

\section{Introduction}\label{sec1}
Distinguishing quantum states forms a fundamental aspect of quantum information processing. In particular, to decode the classical information encoded in a quantum system, it is necessary to distinguish among the possible states of the system by performing measurement(s) on the given quantum system. If the states to be distinguished are nonorthogonal states, then they cannot be distinguished perfectly \cite{Nielsen00}. On the other hand, orthogonal quantum states can always be distinguished perfectly by performing a measurement on the {\it whole} quantum system. In brief, the task of distinguishing quantum states or the state discrimination problem can be defined in the following way: Suppose, a quantum system is given. The system is prepared in a state, taken from a known set. The task is to identify the state by performing measurement(s) on the quantum system.

If a composite quantum system is distributed among several spatially separated parties, then it is usually difficult to perform a measurement on the whole quantum system. In such a situation, it is not always possible to distinguish among the possible states of the system by performing measurements on the local parts of the system and communicating measurement outcomes classically among the parties, even if, to begin with, the states are pairwise orthogonal. Clearly, for composite quantum systems, this class of operations which is commonly known as  local operations and classical communication (LOCC), constitutes a strict subset of all physically realizable operations. Note that in a distributed setting, quantum communication may trivialize a task of distinguishing quantum states, and in any case, such communication is usually difficult in experiments. Thus, we do not consider any quantum communication here, between the spatially separated parties. The question of distinguishability or its absence for pure orthogonal quantum states under LOCC has been considered in several works \cite{Bennett99-1, Bennett99, Walgate00, Ghosh01, Groisman01, Walgate02, Ghosh02, DiVincenzo03, Horodecki03, Badziag03, Ghosh04, Fan04, Hayashi06, Nathanson05, Watrous05, SenDe06, Owari06, Horodecki07, Xin08, Yu11, Bandyopadhyay11, Yu12, Childs13, Cosentino14, Halder19}. If such states are not distinguishable under LOCC, then there are two usual directions to proceed. One is to find an optimal strategy to approximately distinguish those states under LOCC (in this regard, see Refs.~\cite{Virmani01, Chefles00, Jezek02, Eldar03, Bergou04, Barnett09, Bergou10, Weir17, Singal19} for the minimal-error discrimination of nonorthogonal states). Otherwise, those states can be distinguished perfectly under LOCC using suitable resources such as pure state entanglement shared among the spatially separated parties \cite{Cohen07, Cohen08, Bandyopadhyay09-1, Bandyopadhyay10, Bandyopadhyay15, Bandyopadhyay16, Zhang16, Bandyopadhyay18, Zhang18, Halder18, Li19, Rout19, Halder19-2}, or multiple identical copies of the states to be distinguished \cite{Ghosh04, Bandyopadhyay11, Yu14, Li17}. In this work we consider the distinguishability of orthogonal pure states under classes of operations that are more restricted and experimentally friendlier, compared to LOCC, using both types of resources, i.e., shared entanglement in pure form and multiple identical copies of the given states. We also mention that the states to be distinguished are always considered to be equally probable. Note that a particular party can possess more than one qubit/qudit at a time and any type of measurement performed by that particular party on the qubit(s)/qudit(s) in a single location is a local measurement here.

One of the main goals of the present work is to understand the role of classical communication (CC) to distinguish the quantum states when the quantum system is distributed among spatially separated parties. We also remember that performing measurements only after, and depending on the results of, the previous measurements which are performed in another location is experimentally difficult. Here, for most of the paper, we study state distinguishing protocols that involve local projective measurements (LP) without CC (during the measurements) and we say these protocols as LP protocols. The exception is when we consider the {\it incomplete teleportation-based protocol}. The LP protocols are clearly a much more restricted class of operations, compared to LOCC for the state distinguishing problem. We compensate for this restriction separately by two types of resources, viz., shared entanglement and multiple copies of the states to be distinguished. Specifically, we discuss the task of distinguishing quantum states via local projective measurements and shared entanglement (LPSE). In the last couple of decades, there were a few works where local operations and shared entanglement (LOSE) was considered for distinguishing quantum states \cite{Groisman02, Groisman15}. Within a scenario where CC is not allowed during the measurements (we simply say this as without CC), the measurements performed by one party are independent of the measurement outcomes obtained by the other parties. But we remember that once the measurements are accomplished by all the parties, CC is always necessary to identify the state of the system correctly. We also study LP without CC (during the measurements) to distinguish quantum states, provided that a finite number of identical copies of the states to be distinguished are available. Previously, in Refs.~\cite{Ghosh04, Bandyopadhyay11, Yu14, Li17}, the task of distinguishing orthogonal states was considered when multiple identical copies of the given states were available, but there, the operation was LOCC. 

Another aspect of the present work is to explore the (in)distinguishability of quantum states via an incomplete teleportation-based protocol. In a complete teleportation protocol, when an unknown qubit is teleported from one location to another, a two-qubit maximally entangled state (MES) is provided to the spatially separated parties to share and two cbits of information is sufficient to communicate the measurement outcome from one party to another. By an incomplete teleportation-based protocol, we mean that the parties can use an MES but only one cbit of information is allowed, to communicate the measurement outcome from one party to another. In Ref.~\cite{Bandyopadhyay09-1}, the authors used the incomplete teleportation-based protocol, to reduce the average entanglement consumption for distinguishing nonmaximally entangled states (nMESs). For the distinguishability of nMESs see also Ref.~\cite{Gungor16}. 

Here in all the scenarios, CC is limited, i.e., either the parties use CC only after all the measurements are completed, or they use an incomplete teleportation-based protocol, where the availability of CC to communicate the measurement outcome is restricted, in comparison to the standard (complete) teleportation protocol. We now arrange the main findings of the present work as the following: (i) In the problem of distinguishing a specific product basis in $2\otimes2$ (by $m_1\otimes m_2\otimes\dots$, we mean here the Hilbert space $\mathbb{C}^{m_1}\otimes\mathbb{C}^{m_2}\otimes\dots$; however, for simplicity, we consider here all the coefficients as real), we find the relation between the probability of success in minimal-error state discrimination and the entanglement of the resource state (an nMES) via a particular protocol under LPSE. In this regard, we identify a class of nMESs with which it is possible to improve the probability of success compared to that when only LP is allowed. We also calculate the probability of success to distinguish any product basis in $2\otimes2$ using an MES via a similar protocol. In this context, we discuss the status, with respect to the distinguishability of any two or three states in $2\otimes2$ under LPSE. These are contained in Sec.~\ref{sec2subsec1}. We then show that any two-qubit product basis can be perfectly distinguished by LP if two identical copies of the given states are available. This shows the superiority of multiple copies as a resource over shared entanglement, while considering state distinguishability under LP. This result is presented in Sec.~\ref{sec2subsec2}. (ii) In the problem of distinguishing three or four Bell states in $2\otimes2$, within a minimal-error quantum state discrimination, we find the optimal probability under LP while the resource is an nMES. This probability is the same as the probability which is achievable under LOCC or separable measurements, with the resource still being an nMES. This is given in Sec.~\ref{sec3subsec1}. We then show, in Sec.~\ref{sec3subsec2}, that if two of three mutually orthogonal two-qubit states are Bell states, then the trio can be perfectly distinguished by LPSE. (iii) Furthermore, we show, in Sec.~\ref{sec4}, that three two-qubit nonmaximally entangled Bell states (not equally entangled) can be perfectly distinguished by an incomplete teleportation-based protocol. (iv) Then we again go back to the problem of distinguishing quantum states in a multi-copy scenario and show that any product basis can be perfectly distinguished by LP if finite copies of the given states are available. This is given in Sec.~\ref{sec5subsec1}. Thereafter, we compare the preceding result with the distinguishability of entangled states in the multi-copy scenario under LP. This comparison is given in Sec.~\ref{sec5subsec2}.

\section{Discrimination of two-qubit product bases}\label{sec2} 
In this section, we consider the distinguishability of two-qubit product states. Along with LP, additional resources, i.e., shared entanglement or multiple copies of the given states are considered.
\subsection{Discrimination by LPSE}\label{sec2subsec1}
The quantum systems we consider first are two-qubit bipartite systems shared between A(lice) and B(ob).  We begin with a specific two-qubit product basis \cite{Groisman01}, viz.,
\begin{equation}\label{eq1}
\begin{array}{l}
\ket{\psi_1} = \ket{0}\ket{0},~\ket{\psi_2} = \ket{0}\ket{1},~\ket{\psi_3} = \ket{1}\ket{0+1},~\ket{\psi_4} = \ket{1}\ket{0-1},
\end{array}
\end{equation}
where $\ket{v_1+v_2}\equiv(1/\sqrt{2})(\ket{v_1}+\ket{v_2})$ and the symbol $\ket{v_1}\ket{v_2}$ means that $\ket{v_1}$ belongs to Alice and $\ket{v_2}$ belongs to Bob. We use such notations throughout the text and for simplicity we ignore the normalization constants where possible. This basis can be perfectly distinguished by local operations and one-way (from Alice to Bob) classical communication ($1$-LOCC). But if CC is not allowed, then a $2\otimes2$ maximally entangled state is necessary and also sufficient for perfect discrimination of the above basis \cite{Groisman02, Groisman15}. In this sense, it is a ``semilocal'' basis. The discrimination process can be realized by a simple protocol: Alice measures her part in $\{\ket{00}, \ket{01}, \ket{10} \pm \ket{11}\}$ basis and Bob measures his part in the well-known Bell basis $\{\ket{\phi^{\pm}} = (1/\sqrt{2})(\ket{00} \pm \ket{11}), \ket{\psi^{\pm}} = (1/\sqrt{2})(\ket{01} \pm \ket{10})\}$ \cite{Ghosh01}. If Alice and Bob are allowed to perform only LP, and there is no additional resource available, then a maximal success probability of $(1/2 + 1/2\sqrt{2})\approx 0.85$ can be achieved within the scenario of minimal-error quantum state discrimination. The corresponding protocol can be found in \cite{Groisman01}. We have assumed that the states of Eq.~(\ref{eq1}) are provided with equal probabilities. 
\begin{theorem}\label{theo1}
Even a nonmaximally entangled state can improve the success probability in a minimal-error quantum state discrimination to values beyond $0.85$ under LP while distinguishing the states of Eq.~(\ref{eq1}). The probability is a function of the entanglement of the resource state. The ensemble is formed by states of Eq.~(\ref{eq1}) with equal a priori probabilities. 
\end{theorem}
\begin{proof} Suppose, the resource state is $\ket{\phi} = a\ket{0}\ket{0}+b\ket{1}\ket{1}$, with $a,b$ being nonzero positive numbers such that $a>b$ and $a^2+b^2=1$. Appended with the resource state, we rewrite the basis of Eq.~(\ref{eq1}) as 
\begin{equation}\label{eq2}
\begin{array}{c}
a\ket{00}\ket{00} + b\ket{01}\ket{01},~ a\ket{00}\ket{10} + b\ket{01}\ket{11},~a\ket{10}\ket{00} + a\ket{10}\ket{10} + b\ket{11}\ket{01} + b\ket{11}\ket{11}, \\[1 ex] a\ket{10}\ket{00} - a\ket{10}\ket{10} + b\ket{11}\ket{01} - b\ket{11}\ket{11}.
\end{array}
\end{equation}
Notice that the first two states can be perfectly distinguished under the above measurement scheme but for the last two states, there is a possibility of error. It can be shown that there is a connection between the negativity \cite{Zyczkowski98, Lee00, Vidal02, Plenio05} of the resource state and the probability of success in this case. It does not matter which party performs the measurement first. If Bob does it first, then Alice ends up with $a\ket{10}+b\ket{11}$ or $a\ket{10}-b\ket{11}$. If Alice does it first, then Bob ends up with any one of the states $a\ket{00}+b\ket{11}$, $a\ket{00}-b\ket{11}$, $a\ket{10}+b\ket{01}$, or $a\ket{10}-b\ket{11}$. After completing the measurement protocol, the success probability in the whole process is 
\begin{equation}\label{eq3}
\mathcal{P}_{s} = \left(\frac{3}{4} + \frac{ab}{2}\right) = \left(\frac{3}{4} + \frac{1}{2}\mathcal{N}\right),
\end{equation}
where $\mathcal{N} = ab$ is the negativity of the resource state. For $(3/4 + ab/2)>(1/2 + 1/2\sqrt{2})$, we must have $ab>(2-\sqrt{2})/2\sqrt{2}\approx0.2$. 
\end{proof} 
Note that in this section the error is minimum with respect to a particular protocol. Thus, in general the probability of success is not optimal. Let us now move over to the general case. An arbitrary two-qubit product basis can be written as the following (the form is originally given in Ref.~\cite{Walgate02}) 
\begin{equation}\label{eq4}
\begin{array}{c}
\ket{\widetilde{\psi_1}} = \ket{0}\ket{\theta},~\ket{\widetilde{\psi_2}} = \ket{0}\ket{\theta^\prime},~\ket{\widetilde{\psi_3}} = \ket{1}\left(\cos{\frac{\alpha}{2}}\ket{\theta} + \sin{\frac{\alpha}{2}}\ket{\theta^\prime}\right),~\ket{\widetilde{\psi_4}} = \ket{1}\left(\sin{\frac{\alpha}{2}}\ket{\theta} - \cos{\frac{\alpha}{2}}\ket{\theta^\prime}\right),
\end{array}
\end{equation}
where $\ket{\theta}$, $\ket{\theta^\prime}$ are normalized vectors such that $\langle\theta|\theta^\prime\rangle = 0$, and we take $0\leq\alpha\leq\pi/2$. Without loss of generality, we can ignore the relative phase between $\ket{\theta}$ and $\ket{\theta^\prime}$. Still remaining within the framework of minimal-error quantum state discrimination the above basis can be distinguished by LP with a maximum probability of success, $\cos^2{(\alpha/4)}$, when a measurement is performed in the $\{\cos{(\alpha/4)}\ket{\theta} + \sin{(\alpha/4)}\ket{\theta^\prime},~\sin{(\alpha/4)}\ket{\theta} - \cos{(\alpha/4)}\ket{\theta^\prime}\}$ basis by Bob (see the arguments given in Refs.~\cite{Groisman01, Corke17}). A two-qubit MES may not be sufficient to distinguish the above basis under LOSE when $\alpha\neq\pi/2$ \cite{Groisman02, Groisman15}. Nevertheless, we can compute the success probability when one ebit of entanglement is consumed using the following protocol: Alice measures her part in the basis, $\{\ket{00}$, $\ket{01}$, $\ket{10} \pm \ket{11}\}$. Bob measures his part in an entangled basis, $\{\cos{(\alpha^\prime/2)}\ket{\theta0} + \sin{(\alpha^\prime/2)}\ket{\theta^\prime1}$, $\sin{(\alpha^\prime/2)}\ket{\theta0} - \cos{(\alpha^\prime/2)}\ket{\theta^\prime1}$, $\cos{(\alpha^\prime/2)}\ket{\theta1} + \sin{(\alpha^\prime/2)}\ket{\theta^\prime0}$, $\sin{(\alpha^\prime/2)}\ket{\theta1} - \cos{(\alpha^\prime/2)}\ket{\theta^\prime0}\}$, where $0\leq\alpha^\prime\leq\pi/2$. The probability of success in this case is 
\begin{equation}\label{eq5}
\mathcal{P}_s^{(\frac{\alpha}{2}, \frac{\alpha^\prime}{2})} = \frac{1}{2} \left[\sin^2\left({\frac{\alpha+\alpha^\prime}{2}}\right)+\cos^2{\left(\frac{\alpha-\alpha^\prime}{2}\right)}\right].
\end{equation}
For $\alpha = \alpha^\prime = \pi/2$, $\mathcal{P}_s^{(\frac{\pi}{4}, \frac{\pi}{4})} = 1$. However, for a fixed value of $\alpha$, it is better to take $\alpha^\prime = \pi/2$. Keeping this value of $\alpha^\prime$, the above protocol is rather efficient when $\alpha$ is close to $\pi/2$. In particular, for the same value of $\alpha^\prime$ when $\alpha\geq\pi/3$, the above protocol gives an improved probability of success compared to $\cos^2{(\alpha/4)}$. However, there is no known measurement strategy which gives the maximal success probability for any $\alpha$. 

Notice that any three two-qubit pure orthogonal product states can be perfectly distinguished by LP. Without loss of generality, one can choose $\ket{\widetilde{\psi_1}}$, $\ket{\widetilde{\psi_2}}$, $\ket{\widetilde{\psi_3}}$ from Eq.~(\ref{eq4}). In this case, Alice measures in $\{\ket{0}, \ket{1}\}$ basis and Bob measures in $\{\ket{\theta}, \ket{\theta^\prime}\}$ basis. This argument holds true for any set whose constituents are a two-qubit pure entangled state and a two-qubit pure product state, provided the states are orthogonal to each other. Now, two orthogonal pure two-qubit entangled states can be written quite generally (we ignore the phase factors) as
\begin{equation}\label{eq6}
\begin{array}{c}
\ket{\overline{\psi_1}} = a_1\ket{0}\ket{\theta} + a_2\ket{1}\left(\cos{\frac{\alpha}{2}}\ket{\theta} + \sin{\frac{\alpha}{2}}\ket{\theta^\prime}\right),~\ket{\overline{\psi_2}} = a_3\ket{0}\ket{\theta^\prime} +  a_4\ket{1}\left(\sin{\frac{\alpha}{2}}\ket{\theta} - \cos{\frac{\alpha}{2}}\ket{\theta^\prime}\right).
\end{array}
\end{equation}
In the above equation $\ket{\theta}$, $\ket{\theta^\prime}$ are same as defined earlier and $a_1^2+a_2^2 = a_3^2+a_4^2 = 1$, $a_1,~ a_2,~ a_3,~ a_4$ are nonzero positive numbers. Interestingly, distinguishability of the above two states under LPSE is equivalent to the distinguishability of the states of Eq.~(\ref{eq4}). This is due to the arguments, given in \cite{Groisman15}. It is also known that if a set is constituted to a pure two-qubit entangled state and two pure two-qubit product states, then in general such a set can be written as
\begin{equation}\label{eq7}
\begin{array}{c}
\ket{\overline{\psi_1}} = a_1\ket{0}\ket{\theta} + a_2\ket{1}\left(\cos{\frac{\alpha}{2}}\ket{\theta} + \sin{\frac{\alpha}{2}}\ket{\theta^\prime}\right),~\ket{\widetilde{\psi_2}} = \ket{0}\ket{\theta^\prime},\ket{\widetilde{\psi_4}} = \ket{1}\left(\sin{\frac{\alpha}{2}}\ket{\theta} - \cos{\frac{\alpha}{2}}\ket{\theta^\prime}\right).
\end{array}
\end{equation}
In this equation also, $\ket{\theta}$, $\ket{\theta^\prime}$ are same as defined earlier and $a_1^2+a_2^2 = 1$, $a_1, a_2$ are nonzero positive numbers. The above states can be distinguished in the same way as those of Eq.~(\ref{eq4}). Hence, the protocols which are applicable for the states of Eq.~(\ref{eq4}) are also applicable for the states of Eq.~(\ref{eq6}) and Eq.~(\ref{eq7}) when local projective measurements and shared entanglement are allowed. But the success probabilities might be different.

\subsection{Discrimination in the multi-copy scenario}\label{sec2subsec2}
We now consider multi-copy state discrimination scenario for any two-qubit product basis. As stated earlier, in this scenario, multiple identical copies of all the states to be distinguished are available as a resource and the states are to be distinguished under LP. We do not allow here any joint measurement on multiple copies, remembering that a joint measurement is usually a difficult one to implement experimentally. The parties perform measurements on a single copy at a time. Once all the measurements are done on a particular copy then the next copy is provided. In between the copies the parties cannot communicate with each other here. We now present the following theorem. 
\begin{theorem}\label{theo2}
Given any product basis in $2\otimes2$, two identical copies are sufficient to distinguish the states under LP. If $\alpha \neq 0$ in Eq.~(\ref{eq4}), then two copies are also necessary.
\end{theorem}
\begin{proof} Any product basis in $2\otimes 2$ can be written as in Eq.~(\ref{eq4}). We want to distinguish these four states under LP when $\{\ket{\widetilde{\psi_1}}^{\otimes2}, \ket{\widetilde{\psi_2}}^{\otimes2},\ket{\widetilde{\psi_3}}^{\otimes2},\ket{\widetilde{\psi_4}}^{\otimes2}\}$ is available. We remember that the copies are not available at the same time. They are  available one by one. The strategy is simple. Alice can perform only one measurement in $\{\ket{0},\ket{1}\}$ basis. Once this measurement is performed, then two states are eliminated. The information regarding this elimination will be known to both the parties when they will discuss all the measurement outcomes after all measurements are accomplished. Because during the measurements they cannot communicate with each other. The state of system is either from $\ket{\widetilde{\psi_1}}$ and $\ket{\widetilde{\psi_2}}$ or from $\ket{\widetilde{\psi_3}}$ and $\ket{\widetilde{\psi_4}}$. For the first copy Bob can perform the measurement in $\{\ket{\theta},\ket{\theta^\prime}\}$ basis. For the second copy, Bob can perform measurement in $\{\ket{\widetilde{\theta}},\ket{\widetilde{\theta^\prime}}\}$ basis, where $\ket{\widetilde{\theta}}$ = $\cos{(\alpha/2)}\ket{\theta} + \sin{(\alpha/2)}\ket{\theta^\prime}$, $\ket{\widetilde{\theta^\prime}}$ = $\sin{(\alpha/2)}\ket{\theta} - \cos{(\alpha/2)}\ket{\theta^\prime}$. Bob can change his mind and measure in the $\{\ket{\theta},\ket{\theta^\prime}\}$ basis on the second copy and in the $\{\ket{\widetilde{\theta}}, \ket{\widetilde{\theta^\prime}}$ basis on the first copy, without any change in implication. Notice that if $\ket{\widetilde{\theta}}$ can be written as a superposition of both $\{\ket{\theta}$ and $\ket{\theta^\prime}\}$ then two copies are also necessary. Because in this case the states of the set are not perfectly distinguished by local operations only on a single copy of each states \cite{Groisman01}. 
\end{proof}

In Ref.~\cite{Groisman02}, it was shown that the states of Eq.~(\ref{eq4}) cannot always be perfectly distinguished under LOSE, considering a particular type of protocol. Further, in Ref.~\cite{Groisman15}, the optimal amount of entanglement (sometimes many copies of MESs in $2\otimes2$) was derived, that is necessary to transform the states of Eq.~(\ref{eq4}) to $\{\ket{0}\ket{0},\ket{0}\ket{1},\ket{1}\ket{0},\ket{1}\ket{1}\}$ for specific values of ``$\alpha$''. But for certain values of ``$\alpha$'', there is no known protocol via which it is possible to distinguish the states under LOSE (or LPSE), when the resource is a finite number of copies of MESs in $2\otimes2$ shared between the parties. Therefore, we argue here that multiple identical copies of the given states are quite useful compared to shared entanglement, for the task of distinguishing product states in $2\otimes 2$ under LP. However, this may not be true for distinguishing entangled states. A particular such instance is discussed in a later portion, depicting the fact that though the equivalence between the states of Eq.~(\ref{eq4}) and that of Eq.~(\ref{eq6}) was established in Ref.~\cite{Groisman15} when the discrimination of these states is considered under LOSE but for distinguishing the states in the multi-copy scenario under LP, they may not be equivalent.

\section{Discrimination of entangled states by LPSE}\label{sec3}
In this section, we consider the distinguishability of three or four quantum states in $2\otimes2$ under LPSE with a few or all of them being Bell states. The resource state can be a nonmaximally or maximally entangled state.

\subsection{Optimal discrimination of Bell states}\label{sec3subsec1}
We start with the discussion about the distinguishability of ensembles of maximally entangled states in $2\otimes2$, i.e., the Bell states. It is already known that to distinguish three or four Bell states, one ebit of entanglement is necessary and also sufficient under LOCC as well as under separable operations (SEP) \cite{Bandyopadhyay15}. If an nMES of the form $\ket{\phi} = a\ket{0}\ket{0} + b\ket{1}\ket{1}$, where $a,b$ are nonzero positive numbers such that $a>b$ and $a^2+b^2=1$, is supplied, to distinguish three or four Bell states, the optimal probability which can be achieved under SEP or LOCC is derived in Ref.~\cite{Bandyopadhyay15}. Here we show the following.
\begin{theorem}\label{theo3}
The optimal success probability in a minimal-error quantum state discrimination between three or four Bell states under LOCC or SEP with an nMES as a resource can be attained by using LP with the same resource.
\end{theorem}
\begin{proof} Suppose, three Bell states are given along with the resource state $\ket{\phi}$. Then, both the parties can perform a two-outcome projective measurement described by the two projection operators, $\mathbb{P}_1 = \ket{00}\bra{00} + \ket{11}\bra{11}$, $\mathbb{P}_2 = \ket{01}\bra{01} + \ket{10}\bra{10}$, $\mathbb{P}_1+\mathbb{P}_2 = \mathbb{I}$, $\mathbb{I}$ is the identity operator acting on the two-qubit Hilbert space. After completion of the measurements, if the parties get $\mathbb{P}_1\otimes\mathbb{P}_1$ or $\mathbb{P}_2\otimes\mathbb{P}_2$ then the Bell states to be distinguished belong to the subspace spanned by $\{\ket{0}\ket{0}, \ket{1}\ket{1}\}$. This is happening because of the form of the resource state chosen here. Again, if the outcomes are $\mathbb{P}_1\otimes\mathbb{P}_2$ or $\mathbb{P}_2\otimes\mathbb{P}_1$, then the states to be distinguished belong to the subspace spanned by $\{\ket{0}\ket{1}, \ket{1}\ket{0}\}$. In this way, among three Bell states, one state is eliminated as a possible option, or that state is identified as the only option. Again, here the information regarding elimination or identification will be known to both the parties when they discuss the measurement outcomes after all the measurements are accomplished. If one state gets eliminated, then the remaining two states of the ensemble becomes nonorthogonal to each other, optimal discrimination of which can be done when both the parties perform measurements in Bell basis, leading to a probability $(1/2+ab)$ of success. Therefore, average optimal probability to distinguish three Bell states is $(1/3)+(2/3)(1/2 + ab)$ = $2/3 + (2/3)ab$, which is the same as obtained before by LOCC (or SEP).

For four Bell states, both the parties avail the same protocol as described above. But in this case, no states can be distinguished straightaway with some probability. After performing the measurements by both the parties, described by $\mathbb{P}_1$ and $\mathbb{P}_2$, two states get eliminated. So, the remaining states can be distinguished with a probability of $(1/2 + ab)$. This is also the optimal probability, as obtained before by LOCC (or SEP). \end{proof}
Notice that $\mathcal{N} = ab$ is the negativity of the resource state $\ket{\phi}$. Therefore, in this case also the success probability is a function of the entanglement content of the resource state.
\subsection{Perfect discrimination of three quantum states}\label{sec3subsec2}
The discussion in the preceding section makes it clear that when an entangled resource state, $\ket{\phi}$, is available, it is possible to distinguish between the subspaces spanned by $\{\ket{00}, \ket{11}\}$, $\{\ket{01}, \ket{10}\}$, by performing the two-outcome projective measurement onto $\mathbb{P}_1$ and $\mathbb{P}_2$. Thus, we arrive to the following theorem for distinguishing states under LPSE.
\begin{theorem}\label{theo4}
Consider three mutually orthogonal quantum states in $2\otimes2$. Two of them are Bell states, and the third is any state orthogonal to the two Bell states. These states can be perfectly distinguished under LPSE. 
\end{theorem}
\begin{proof}Recall that in the protocol considered in Theorem \ref{theo3}, we described that a state can be eliminated by distinguishing the subspaces spanned by $\{\ket{0}\ket{0}, \ket{1}\ket{1}\}$ and $\{\ket{0}\ket{1}, \ket{1}\ket{0}\}$, when an entangled resource state is available. If the resource state is a Bell state ($\frac{1}{\sqrt{2}}\ket{0}\ket{0} + \frac{1}{\sqrt{2}}\ket{1}\ket{1}$), then nonorthogonality does not occur. So, perfect discrimination is possible by performing measurements in the Bell basis by both the parties. As mentioned in the Theorem \ref{theo4}, three states can be described as $\{\ket{0}\ket{0} \pm \ket{1}\ket{1}, a\ket{0}\ket{1} + b\ket{1}\ket{0}\}$, $\{\ket{0}\ket{1} \pm \ket{1}\ket{0}, a\ket{0}\ket{0} + b\ket{1}\ket{1}\}$ or $\{\ket{0}\ket{0} + \ket{1}\ket{1}, \ket{0}\ket{1} + \ket{1}\ket{0}, a(\ket{0}\ket{0} - \ket{1}\ket{1}) + b(\ket{0}\ket{1} - \ket{1}\ket{0})\}$ (normalization ignored), $a,b$ can be any positive numbers such that $a^2+b^2=1$. Notice that the distinguishability of the first two ensembles can be accomplished by the above protocol. For the last ensemble, the distinguishability of the states is the same as the distinguishability of four Bell states under LPSE. 
\end{proof}

\section{Perfect discrimination of states via the incomplete teleportation-based protocol}\label{sec4}
We now discuss the incomplete teleportation-based protocol (iCTP) to distinguish nonmaximally entangled states perfectly. As mentioned earlier, in such a protocol, a Bell state as resource will be provided to the parties for distinguishing the given ensemble. But they are restricted to communicate only one cbit of information. Using these resources, the parties complete the task of distinguishing the given states perfectly. It is possible to show that there exist four nMESs which can be perfectly distinguished by an iCTP. These states can be found from Eq.~(\ref{eq8}) by putting $c=a$ and $d=b$. In this case, the states to be distinguished, are equally entangled. We now show that even if the states are not equally entangled, any three states chosen from the following basis, can be distinguished perfectly by iCTP. We first give the form of four orthogonal nMESs, which are not equally entangled:
\begin{equation}\label{eq8}
\begin{array}{c}
\ket{\phi_1} = a\ket{0}\ket{0}+b\ket{1}\ket{1},~
\ket{\phi_2} = b\ket{0}\ket{0}-a\ket{1}\ket{1},~
\ket{\phi_3} = c\ket{0}\ket{1}+d\ket{1}\ket{0},~
\ket{\phi_4} = d\ket{0}\ket{1}-c\ket{1}\ket{0},~
\end{array}
\end{equation} 
where $a,b,c,d$ are nonzero real numbers such that $a>b$, $c>d$, $a^2+b^2$ = $1$ = $c^2+d^2$, $a\neq c \neq d$. We consider any three states of the above form. Note that which three states are chosen, is known to the parties. Using an entangled resource, it is always possible to eliminate one state (see the method in the proof of Theorem \ref{theo3}). Now if the parties use an MES as a resource, then the remaining states are orthogonal. To distinguish between them, the following strategy can be fixed: Alice performs measurement in the Bell basis and then informs Bob whether she got ``$+$'' or ``$-$'', i.e., whether she got a click in the $\{\ket{00}+\ket{11}, \ket{01} + \ket{10}\}$ space or in the complementary one. Based on that, Bob can choose his basis. However, this method does not work for all four states given in the above equation. It is actually not known whether the above four states can be perfectly distinguished by iCTP.

\section{Discrimination of quantum states in the multi-copy scenario}\label{sec5} 
In this section, we again go back to the state discrimination problem in multi-copy scenario under LP. We first talk about the product state discrimination for any arbitrary Hilbert space and then compare the result with the distinguishability of entangled states under the same setting. 
\subsection{Perfect discrimination of product states}\label{sec5subsec1}
We consider multipartite orthogonal product bases in $d_1\otimes d_2\otimes\cdots\otimes d_m$, where $m$ is the number of parties and $d_i$ is the dimension of $i$-th subsystem. For such a basis, even if it is indistinguishable under LOCC, it can be perfectly distinguished under LP, provided multiple identical copies of the basis states are available. As mentioned earlier, these copies are given one by one. So, measurement on only a single copy at a time is allowed. 
\begin{theorem}\label{theo5}
To distinguish any orthogonal product basis under LP, a finite number of copies of the given states is sufficient.
\end{theorem}
\begin{proof} Within an $m$-partite basis in $d_1\otimes d_2\otimes\cdots\otimes d_m$, the total number of states is $d_1d_2\dots d_m$. For product states, they maintain orthogonality from at least any one party's side. In this way, when the full basis is formed, orthogonality ``saturates'', that is, there are no other states which are again orthogonal to all the states of the given basis. Therefore, once the first copy is given to the $m$ parties, the first party can choose a basis to perform a $d_1$-outcome projective measurement. By this measurement, the first party can always eliminate $d_1-1$ states. In this way, $i$-th party can perform a $d_i$-outcome projective measurement to eliminate $d_i-1$ states. So, using the first copy all the parties perform measurements and eliminate at least $\sum_i(d_i-1)$ states. The remaining states are $\Pi_i d_i - \sum_i(d_i-1)$. The second copy onward, it may not be possible to eliminate so many states. But at least one state can always be eliminated by performing a projective measurement by every party, and thereby for each copy (after the first copy), $m$ states can be eliminated. Let the smallest integer be $x$ which is greater than or equal to $[\Pi_i d_i - \sum_i(d_i-1)]/m$. Therefore, the total number of copies required to complete the task is $(x+1)$. Note that using a particular copy it may possible to eliminate more states than what we have indicated in the above. In such a situation for later copies the number of states which can be eliminated, can be reduced. This is simply because there are not too many states left for the elimination procedure.
\end{proof}
This is an improved bound with respect to the bound derived in Ref.~\cite{Walgate00}. Though we have considered a restricted class of states (complete orthogonal product basis), but we have achieved the present bound under a very restricted class of operations, viz. LP, as compared to LOCC in Ref.~\cite{Walgate00}.

\subsection{Comparison with entangled states}\label{sec5subsec2}
The strategy that we adopted in the preceding section may not work for entangled states even if there are only two orthogonal states. Consider the problem of distinguishability of the two states,
\begin{equation}\label{eq9}
a\ket{0}\ket{0} + b\ket{1}\ket{1},~ b\ket{0}\ket{0} - a\ket{1}\ket{1},
\end{equation}
under LP, and when multiple copies - provided one by one - are available. In the above equation, $a,b$ are nonzero real numbers such that $a>b$, $a^2+b^2$ = $1$. Consider the first copy provided to the parties. As long as projective measurements are used, the party who performs the measurement first can only measure in the eigenbasis of $\sigma_x$, i.e., in the $\ket{\pm} = (1/\sqrt{2})(\ket{0}\pm\ket{1})$ basis to keep the post-measurement states orthogonal. We consider here those measurements which preserves orthogonality at the other side. But depending on the two different measurement outcomes, the sets of orthogonal states which are produced on Bob's side are different. Bob may get $\{a\ket{0} + b\ket{1}, b\ket{0} - a\ket{1}\}$ or $\{a\ket{0} - b\ket{1}, b\ket{0} + a\ket{1}\}$. However, no communication is allowed from Alice's side to Bob. In this situation, Bob has to choose his measurement randomly. Without loss of generality Bob can decide that on the first copy he measures in $\{a\ket{0}+b\ket{1}, b\ket{0} - a\ket{1}\}$ and for the second copy he measures in $\{a\ket{0}-b\ket{1}, b\ket{0} + a\ket{1}\}$. But the problem occurs if for the first copy Alice gets ``$-$'' and for the second copy Alice gets ``$+$''. In fact, this situation may arise for every copy provided to the parties. So, this is a type of nonlocality (here it means, indistinguishability of orthogonal states under a particular class of operations acting locally) exhibited by the entangled states under the LP class of operations. Local indistinguishability in the many-copy scenario was studied earlier in \cite{Bandyopadhyay11, Yu14, Li17}. But there, the allowed operations were much stronger than LP. Thus, they considered at least a mixed state in the ensemble, because any $N$ pure states can be perfectly distinguished by LOCC if $N-1$ copies of each state are available \cite{Walgate00, Bandyopadhyay11}. 

Notice that for a particular form of two orthogonal states, which can be found by putting $\alpha = \pi/2$ in Eq.~(\ref{eq6}), they can be perfectly distinguished by LPSE when one ebit of entanglement is present as resource (for LOSE see Ref.~\cite{Groisman15}). But from the above discussion, it is clear that these two states can not be distinguished by LP in the many-copy scenario. In Section \ref{sec3}, we had identified a scenario where LPSE turned out to be weaker in comparison to LP assisted with multiple copies of the ensemble states. The tables, however, have turned here, and we have now an example where the opposite is true. Moreover, the result also implies that in the many-copy scenario, (even if $\alpha = \pi/2$), the states of Eq.~(\ref{eq6}) are not equivalent to that of Eq.~(\ref{eq4}), with respect to discrimination under LP, though the equivalence is known to be true under LOSE \cite{Groisman15}.

\section{Conclusion}\label{sec6}
We have considered  different tasks of quantum state discrimination under LP in limited communication scenarios. We have tried to shed light on the role of CC to accomplish different state discrimination tasks. In most of the cases, we have considered no CC (during the measurements), and instead have allowed the utilization of other resources, viz., a shared entangled state or multiple copies of the ensemble states to be distinguished. We have also compared the different types of resources where possible. It is observed that if multiple identical copies of the given states are available, then it is a quite effective resource for distinguishing product states under LP. However, for distinguishing entangled states, entanglement can be an effective resource under LP. We find that ensembles of entangled states can exhibit an indistinguishability property in the many-copy scenario under LP which never occurs for product states.

\section*{Acknowledgment}
The authors are grateful to Ujjwal Sen for his valuable comments, based on which the paper was improved significantly.

\bibliography{ref}
\end{document}